\documentclass[a4paper,english]{lipics-v2021}
\hideLIPIcs

\usepackage{amsthm}
\usepackage{amsmath}
\usepackage{amssymb}
\usepackage{amsfonts}
\usepackage{epsfig}
\usepackage{graphics}
\usepackage{graphicx}
\usepackage{booktabs}
\usepackage{pdfpages}
\usepackage{MnSymbol}

\nolinenumbers

\usepackage{complexity}
\usepackage{url}

\usepackage{latexsym}
\usepackage{psfrag}
\usepackage{enumerate}
\usepackage{here}
\usepackage{esvect}
\usepackage[update,prepend]{epstopdf}
\usepackage{anyfontsize}
\usepackage{url}
\usepackage{hyperref}

\usepackage{microtype}

\newtheorem{problem}[theorem]{Problem}

\newcommand{\etal}{{et~al.}}
\newcommand{\ie}{{i.e.}}

\newcommand{\eps}{\varepsilon}

\providecommand{\intd}[0]%
{\;\mbox{d}}

\usepackage[noline,linesnumbered]{algorithm2e}
\SetArgSty{textrm}
\let\oldnl\nl
\newcommand{\nonl}{\renewcommand{\nl}{\let\nl\oldnl}}
\makeatletter
\def\TitleOfAlgo{\@ifnextchar({\@TitleOfAlgoAndComment}{\@TitleOfAlgoNoComment}}
\def\@TitleOfAlgoAndComment(#1)#2{\nonl\hspace*{-1.5em}#2 #1\;}
\def\@TitleOfAlgoNoComment#1{\nonl\hspace*{-1.5em}#1\;}
\makeatother
\DontPrintSemicolon

\newcommand{\later}[1]{}
\newcommand{\old}[1]{}

\title{A Purely Geometric Variant of the Gale--Berlekamp Switching Game}

\titlerunning{A Purely Geometric Variant of the Gale--Berlekamp Switching Game}

\author{Adrian Dumitrescu}
{Algoresearch L.L.C., Milwaukee, WI, USA}
{ad.dumitrescu@algoresearch.org}
{0000-0002-1118-0321}
{}

\author{Jeck Lim}
{California Institute of Technology, Pasadena, CA, USA}
{jlim@caltech.edu}
{}
{}

\author{J\'anos Pach}
{Alfr\'ed R\'enyi Institute of Mathematics, Budapest, Hungary}
{pach@renyi.hu}
{0000-0002-2389-2035}
{}

\author{Ji Zeng}
{Alfr\'ed R\'enyi Institute of Mathematics, Budapest, Hungary}
{jzeng@ucsd.edu}
{}
{}

\authorrunning{Adrian Dumitrescu, Jeck Lim, J\'anos Pach, and Ji Zeng}

\Copyright{Adrian Dumitrescu, Jeck Lim, J\'anos Pach, and Ji Zeng}

\funding{Work on this paper was supported by ERC Advanced Grant
  no. 882971  ``GeoScape''. Work by the second named author was
  partially supported by an NUS Overseas Graduate Scholarship while in
  residence at the Simons--Laufer Mathematical Sciences Institute
  during the Spring 2025 semester. Work by the third named author was
  also supported by the Eisenbud Professorship and and by a Clay Senior 
  Scholar fellowship at the IAS/Park City Mathematics Institute in July 2025. 
  Work by the fourth named author was also supported by
  ERC Advanced Grant no. 101054936 ``ERMiD'', and by NSF grant
  DMS-1928930 while in residence at the Simons--Laufer Mathematical
  Sciences Institute during the Spring 2025 semester.} 



\begin{document}

\maketitle

\begin{abstract}
We introduce the following variant of the Gale--Berlekamp switching game.
Let $P$ be a set of n noncollinear points in the plane, each of them having weight $+1$ or $-1$.
At each step, we pick a line $\ell$ passing through at least two points of $P$, and switch
the sign of every point $p \in P\cap\ell$. The objective is to
maximize the total weight of the elements of $P$. We show that one can
always achieve that this quantity is at least $n - o(n)$, as
$n\rightarrow\infty$, and at least $n/3$, for every $n$. Moreover,
these can be attained by a polynomial time algorithm. 
\end{abstract}

\section{Introduction} \label{sec:intro}

The Gale--Berlekamp switching game (introduced in the $1960$s)  can be described as follows.
Consider a $\sqrt{n} \times \sqrt{n}$ array of $n$ lights. Each light has two possible states,
\emph{on} (or $+1$) and \emph{off} ($-1$, respectively).
To each row and to each column of the array there is a switch. Turning a switch changes the state of
each light in that row or column. Given an initial state of the board, \ie, a certain on or off position
for each light in the array, the goal is to turn on as many lights as possible.
Equivalently, the number of lights that are on minus the number of lights that are off
should be as large as possible, when starting from the initial configuration. See Fig~\ref{fig:board}.
The first correct analysis of the original version of the game with $n=100$ was given by
Carlson and Stolarski~\cite{CaS} in 2004. (cf.~\cite{FS89}.)

\begin{figure}[htbp]
 \centering
 \includegraphics[scale=0.95]{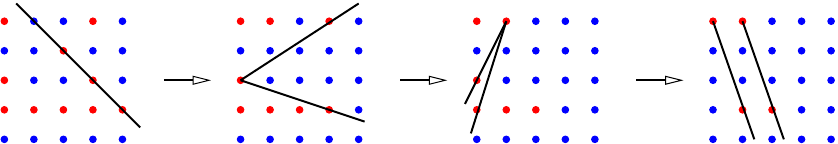}
 \caption{A $5 \times 5$ grid for the original game whose signed discrepancy is $5$ and whose
   maximum signed discrepancy is $15$ (left); points with weight $-1$ are red.
   The latter parameter becomes $25$ in the purely geometric variant, where one switches
   all lights along any line (not necessarily vertical or horizontal) spanned by the points.} 
 \label{fig:board}
\end{figure}

By a \emph{configuration} we shall mean the state of the machine, \ie, precisely which lights are on and off.
Following Beck and Spencer~\cite{BS83}, define the \emph{signed discrepancy} of a configuration
as the number of lights that are on minus the number of lights that are off.
Equivalently, this is just the sum of the weights $+1$ and $-1$ corresponding to the on and off lights,
respectively.

The asymptotic behavior of the original version of the game, played on the $\sqrt{n}\times\sqrt{n}$
square board, was studied by Brown and Spencer~\cite{BS71} and revisited by Spencer~\cite{Spe94}.
They showed that one can always achieve a signed discrepancy of $\Omega(n^{3/4})$ and that
the order of magnitude of this bound cannot be improved. Interestingly enough, both the upper
and the lower bound are obtained by the probabilistic method~\cite{AS16}.
\smallskip

Here we consider the following, purely geometric, variant of the switching game.
The board is any noncollinear $n$-element point set $P$ in the plane,  equipped with a switch
for each connecting line. A \emph{configuration} is a pair $(P,w)$, where $w$ is an assignment
of $\pm 1$  weights to the elements of $P$. Denote by $|w|$ the sum of the weights $\sum_{p\in P} w(p)$.

Let $F(P,w_0)$ denote the largest signed discrepancy that can be achieved by applying a series
of switches to the \emph{initial configuration} $(P,w_0)$, that is, let
\begin{equation} \label{eq:initial}
  F(P,w_0) = \max_{w} |w|,
\end{equation}
where the maximum is taken over all weight assignments $w : P \rightarrow\{+1,-1\}$ that can be obtained
starting with the initial assignment $w_0$. Finally, define 
\begin{equation} \label{eq:F}
  F(P) = \min_{w_0} F(P,w_0),
\end{equation}
where $w_0$ runs through all initial weight assignments.

We show that, regardless of the initial assignment, if $n$ is large
enough, then a skilled player can always achieve a linear signed
discrepancy $cn$, for any constant $c$ arbitrarily close to 1 and
independent of $P$. Further, one could do this by only applying a
linear number of switches out of possibly a quadratic number. 

\begin{theorem}\label{thm:new}
  For any noncollinear $n$-element point set $P$ in the plane, consider a board where there is
  a switch for each connecting line.
  
  Then we have $F(P) \geq n - o(n)$, i.e., in the corresponding game,
  one can always turn on at least $n-o(n)$ lights. Moreover, the lower bound can be achieved
  by an application of $O(n)$ switches.
\end{theorem}

Our next theorem below gives a concrete bound that holds for every $n$.

\begin{theorem} \label{thm:old}
  For any noncollinear  $n$-element point set $P$ in the plane, consider a board where there is
  a switch for each connecting line.
  
  Then we have $F(P) \geq n/3$, i.e., in the corresponding game,
  one can always turn on at least $2n/3-1$ lights. Moreover, the lower bound can be achieved
  by an application of $O(n)$ switches.
\end{theorem}

Notice that if $n$ is odd and all but one of the points of $P$ are collinear,
then we have $F(P) \leq n-2$. Indeed, if initially an odd number of lights are off $(-1)$,
then one cannot turn all lights on, because no switch can change the parity of the off lights.
Actually, it is easy to see that for this board, we have $F(P)=n-2$.
On the other hand, if $n$ is even and all but one of the points of $P$ are collinear,
then one can turn on all the lights. See Fig.~\ref{fig:col}.

\begin{figure}[htbp]
\centering
\includegraphics[scale=0.95]{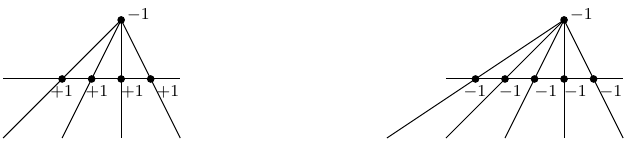}
\caption{Two warm-up board examples.}
\label{fig:col}
\end{figure}

Our proofs crucially rely on the concept of \emph{ordinary line}.
Given a point set $P$, a connecting line is called \emph{ordinary} if it contains precisely
two points of $P$~\cite{BM90}. 

It will be evident from the proof of Theorem~\ref{thm:old} that
repeatedly computing a line incident to the maximum number of points
and reducing the current size by $3$ yields a cubic-time
algorithm. The proof of Theorem~\ref{thm:new} uses more technical
machinery but repeatedly computing the ``ordinary line graph'' of the
point set under consideration yields a polynomial-time algorithm that
achieves $F(P) \geq (1-\eps)n$ for $n$ sufficiently large
depending on arbitrary $\eps > 0$. 

Some interesting questions remain:

\begin{problem}
  Given a noncollinear $n$-element point set in the plane with assigned $\pm 1$ weights,
  can one efficiently compute a sequence of switches that produce a weight assignment
  of maximum signed discrepancy?
\end{problem}

\begin{problem}
    Is there an absolute constant $c$ such that given any plane
    noncollinear $n$-element point set $P$, one can achieve the
    maximum signed discrepancy $F(P) \geq n - c$? 
\end{problem}

\section{Proof of Theorem~\ref{thm:new}}

To prove Theorem~\ref{thm:new}, we introduce the notion of
\textit{ordinary line graph}. The ordinary line graph associated with
a point set $P$ is defined as an auxiliary graph $G$ whose vertex set
$V(G)$ consists of all points in $P$, and two points are connected by
an edge if and only if they determine an ordinary line.
The following lemma will be handy for us. 

\begin{lemma}\label{lem:component}
Let $G$ be the ordinary line graph associated to a point set $P$ in
the plane, and $u$ be a vertex inside a connected component $H$ of
$G$. 

Then any configuration $(P,w)$ can be transformed into another
configuration $(P,w')$ through a series of line switches such that
$w'(v) = +1$ for all $v$ in $H$ except $u$ (i.e., $w'(u)$ could be
either $\pm 1$), and $w'(v) = w(v)$ for all $v$ not in
$H$. Furthermore, the number of line switches needed is at most
$|V(H)|-1$.  
\end{lemma}

\begin{proof}
Consider a spanning tree $T\subset H$ with $u$ being its root. For each vertex
$v \in V(T)$, its depth is the length of the path from $v$ to the root $u$. Let
$D$ be the maximum depth of this tree. Iteratively for $i = D, D-1, \dots, 1$,
we apply line switches to each $v \in V(T)$ of depth $i$, using the ordinary
line represented by the edge between $v$ and its parent, to switch the weight of
$v$ to be $+1$. Obviously, at the $i$-th round, only the weights of vertices of
depths $i$ or $i-1$ may be affected. Observe that we only ever apply line
switches associated to the edges of $T$, so we have applied at most
$|E(T)|=|V(T)|-1$ many switches. 
\end{proof}

Another ingredient of our proof is the following seminal result of Green and Tao~\cite[Thm.~1.5]{GT13}.
We give the full statement of this result involving terminologies from algebraic geometry.
However, as we shall explain, we only need a weaker consequence of the conclusion.

\begin{theorem}[Green-Tao]\label{thm:greentao}
Suppose $P$ is a finite set of $n$ points in the real projective plane
$\mathbb{RP}^2$. Let $K > 0$ be a real parameter. Suppose that $P$ spans at most
$Kn$ ordinary lines. Suppose also that $n \geq \exp\exp(CK^C)$ for some
sufficiently large absolute constant $C$. 

Then, after applying a projective transformation if necessary, $P$ differs by at
most $O(K)$ points (which can be added or deleted) from an example of one of the
following three types:
\begin{itemize}
    \item[I.] $n - O(K)$ points on a line.
    \item[II.] For some $m = n/2 + O(K)$, the set
      \begin{equation*}
            X_{2m} = \left\{[\cos{\frac{2\pi j }{m}},\sin{\frac{2\pi j}{m}},1]:
            0 \leq j < m\right\} \cup \left\{[-\sin{\frac{\pi j}{m}},
              \cos{\frac{\pi j}{m}}, 0]: 0 \leq j < m\right\}. 
        \end{equation*}
    \item[III.] A coset $H + h$, where $3h \in H$, of a finite subgroup $H$ of
      the nonsingular real points on an irreducible cubic curve, with $H$ having
      cardinality $n + O(K)$. 
\end{itemize}
\end{theorem}

Here, a cubic curve refers to the zero set of a homogeneous cubic polynomial,
and it is irreducible if its defining cubic polynomial is irreducible.
It is well-known in the study of elliptic curves that the nonsingular real points
on an irreducible cubic curve form an abelian group.
However, we will not need these notions in our proof of Theorem~\ref{thm:new},
although the curious readers are referred to Section~2 of~\cite{GT13}.
It is only important to us that a line and a cubic curve intersect in at most three points.
The following corollary is all we need from the above theorem.

\begin{corollary} \label{cor:atmost1triple}
Suppose $P$ is a finite set of $n$ points in the real projective plane
$\mathbb{RP}^2$ satisfying the assumptions in Theorem~\ref{thm:greentao}. Then,
after applying a projective transformation and deleting at most $O(K)$ points from $P$ to obtain $Q$,
there is at most one line passing through more than 3 points of $Q$.
\end{corollary}
\begin{proof}
  After applying Theorem~\ref{thm:greentao} and adding and/or deleting $O(K)$ points from $P$,
  we obtain $Q'$ satisfying one of the types in Theorem~\ref{thm:greentao}.
  Take $Q=Q'\cap P$, i.e., $Q$ is obtained by deleting $O(K)$ points from $P$,
  without adding any point. Then $Q$ is a subset of one of the three types.
    
    In Type I, $Q$ consists of a set of collinear
    points, thus there is at most one line through more than 3 points. In Type II, $Q$ consists of a
    set of collinear points (on the line at infinity) and another set of points
    on the unit circle, so every line, except for the line at infinity, passes through at most 3 points.
    In Type III, any line intersects $Q$ in at most three
    points because $Q$ is on an irreducible cubic curve. 
\end{proof}

\begin{remark*}
  A weaker structure theorem~\cite[Thm.~1.4]{GT13} to Theorem~\ref{thm:greentao} states that
  almost all points of $P$ lie on a cubic curve. This is insufficient to deduce
  Corollary~\ref{cor:atmost1triple} since the cubic curve could be a union of three lines.
\end{remark*}

Now we are ready to prove Theorem~\ref{thm:new} which follows from the next statement.

\begin{proposition}
    For any real number $\eps > 0$, there exists a constant $n_\eps$
    such that any initial configuration $(P,w_0)$ on a noncollinear point set
    $P$ in the plane can be transformed into a configuration $(P,w)$ such that
    $|w| \geq (1- \eps) |P| - n_\eps$ after $O(|P|)$ line switches. 
\end{proposition}

\begin{proof}
    The proof is by induction on $|P|$ denoted as $n$. For $n < n_\eps$
    where $n_\eps$ is to be chosen later, this statement is vacuously
    true. With foresight, we choose $K = \lceil 2/ \eps \rceil$. Let $G$ be
    the ordinary line graph associated with $P$, and let $H$ be the largest component of $G$.
    Our proof is then divided into two cases,
    depending on whether the size of $H$ is big or small.

    \smallskip
    \emph{Case 1}: $H$ has size at
    least $K$. If $P \setminus V(H)$ is not collinear, we can apply the
    inductive hypothesis to conclude that after $O(|P| - |V(H)|)$ line switches,
    the sum of weights of $P \setminus V(H)$ is at least $(1 - \eps)(|P| - |V(H)|) - n_\eps$.
    If $P \setminus V(H)$ is collinear, there exists $p
    \in P$ not on the unique line determined by $P \setminus V(H)$. Applying
    switches using lines determined by $p$ and points in $P \setminus V(H)$, the
    sum of weights of $P \setminus V(H)$ can be $|P| - |V(H)|$, which is even
    better than the previous case. 
    
    Now we can apply Lemma~\ref{lem:component} to
    guarantee that all but at most one point in $H$ are assigned with $+1$
    without affecting any points in $P \setminus V(H)$ using at most $|V(H)|$
    line switches. By the choice of $K$, the sum of weights of $H$ is at least
    $|V(H)|-2\geq (1 - \eps)|V(H)|$, hence the sum of weights of $P$ is at
    least 
    \[ (1 - \eps)(|P| - |V(H)|) - n_\eps + (1 - \eps)|V(H)|=(1 -\eps)|P| - n_\eps. \] 

    \smallskip
    \emph{Case 2}: $H$ has size less than
    $K$. Then, each vertex in $G$ has degree less than $K$, so the number of
    edges of $G$ is $< Kn$. Hence, $P$ determines fewer than $Kn$ ordinary
    lines. We choose $n_\eps \geq \exp \exp (CK^C)$ for $C$ as in
    Theorem~\ref{thm:greentao}. For $n\geq n_\eps$,
    Corollary~\ref{cor:atmost1triple} implies that, after a projective transformation and deleting
    at most $\delta=O(K)$ points of $P \subset \mathbb{R}^2 \subset \mathbb{R}\mathbb{P}^2$,
    we obtain a set $Q$ satisfying at most one line passes through more than three points of $Q$.
    If $Q$ is contained in a line, then we add any point of $P$ which does not lie on that line back to $Q$.
    Thus we may assume that $Q$ is noncollinear.
    Furthermore, we can choose $n_\eps \geq 2\delta + 2$, which implies that $|Q| - 2 - \delta \geq (1 - \eps)|P| - n_\eps$.
    We conclude the inductive process using the following claim.  
    
    \begin{claim}\label{cubic}
        Any initial configuration $(Q,w_0)$ on a noncollinear point set $Q$,
        where at most one line intersects $Q$ in more than three points, can be
        transformed into a configuration $(Q,w)$ with $|w| \geq |Q| - 2$ after
        $O(|Q|)$ line switches. 
    \end{claim}

    The rest of the proof is devoted to this claim. Let $G$ denote the ordinary
    line graph associated with the given $Q$. Our proof of the claim is by
    induction on the number of connected components of $G$. In the base case,
    $G$ has only one component and Lemma~\ref{lem:component} gives us what we
    want. 
    
    For the inductive process, we can take the largest component $H$ of
    $G$. Notice that $|V(H)| \geq 2$ is guaranteed by the Sylvester--Gallai
    theorem (see, for example, \cite{GT13} Theorem~1.1). Also, there exists a
    point $u \in V(H)$ not on the potential line intersecting $Q$ in more than
    three points (otherwise there won't be any edges within $H$). If $Q\setminus V(H)$ is noncollinear,
    we can perform $O(|Q|- |V(H)|)$ line switches to make
    sure that at most one point in $Q \setminus V(H)$ is assigned weight $-1$ by
    the inductive hypothesis. If $Q\setminus V(H)$ is a collinear set, we can
    achieve the same effect by switching the lines passing through a point in
    $V(H)$ not collinear with $Q\setminus V(H)$. Then, we can apply
    Lemma~\ref{lem:component} to make sure that all points in $H$ other than $u$
    are assigned weight $+1$. Now at most two points in $Q$ could be assigned
    weight $-1$, and we only need to deal with the case where there are exactly
    two such points. In that case, these two points are $u$ and some $v\in Q\setminus V(H)$.
    Switch the line through $u,v$. Because this line contains
    the point $u$, it only intersects $Q$ in at most three points. Hence this
    switch reduces the number  of $-1$'s assigned to $Q$ to at most one, concluding the proof. 
    \end{proof}

\section{Proof of Theorem~\ref{thm:old}}

A line is said to be \emph{negative} (resp., \emph{positive}, \emph{zero}, \emph{nonnegative}),
if the sum of the weights of the points incident to it is as such.
A point set is said to be \emph{negative} (resp., \emph{positive}, \emph{zero}, \emph{nonnegative}),
if the sum of the weights of its points is as such.
A point pair $p,q$ is said to be a \emph{negative pair} if $w(p)=w(q)=-1$.
(The last two definitions are consistent.)

\smallskip
Throughout the proof we will use the following procedure.

\subparagraph{Procedure \texttt{N}.}
Given an input point set, repeatedly switch negative connecting lines one by one in an arbitrary order.
Note that this procedure terminates, since the number of positive weights strictly increases by $2$
at each application. Observe that \texttt{N} terminates with every connecting line of the point set
being nonnegative.

\begin{lemma} \label{lem:gen-pos}
  Let $P$ be any set of $n$ points in general position in the plane.
  Then the maximum discrepancy of the corresponding board is at least $n-2$.
  This inequality is the best possible.
\end{lemma}
\begin{proof}
  First use Procedure \texttt{N} on input $P$. That is, while there exists a negative pair,
  switch the corresponding line. When the number of negative pairs is zero,
  the number of $-1$ weights is at most one; which yields the claimed bound.
  A set of $n-1$ weights $+1$ and one  weight $-1$ placed arbitrarily on the $n$ points
  provides a tight bound; indeed, the odd parity of the total negative weight is maintained
  after each switch.
\end{proof}

The next two lemmas implement the induction step and the induction basis, respectively.

\begin{lemma} \label{lem:reduction}
  Let $P$ be any set of $n \geq 3$ noncollinear points in the plane.
  Let $p,q,r \in P$ so that $P'= P \setminus \{p,q,r\}$ is not collinear, and either
  \begin{enumerate} [(i)] \itemsep 1pt
  \item $p,q,r$ are collinear on a line incident to exactly three points; or
  \item $\ell_1= \ell(p,q)$ and $\ell_2=\ell(p,r)$ are ordinary lines sharing a common point.
  \end{enumerate}
  Then $F(P) \geq F(P') +1$.
\end{lemma}
\begin{proof}
Let $T=\{p,q,r\}$.
After a suitable sequence of line switches determined by $P'$, we obtain a
weight assignment $w$ on $P$ such that $\sum_{p\in P'} w(p)=F(P)$. Note that
such switches may affect $T$, as points of $T$ may be incident to lines
determined by $P'$. 

(i)~Assume first that $p,q,r$ are collinear on $\ell$. Switch $\ell$ if $\ell$ is negative,
to obtain a weight assignment $w'$.
We have $\sum_{x \in T} w'(x) \geq 1$, thus $F(P) \geq F(P') +1$ since $w$ and $w'$ agree on $P'$.

(ii)~Assume now that $\ell_1= \ell(p,q)$ and $\ell_2=\ell(p,r)$ are ordinary lines sharing a common point.
Switch $\ell_1$ and/or $\ell_2$ if they are negative, to obtain a weight assignment $w'$. 
If $w'(p)=-1$, $\sum_{x \in T} w'(x) =1$, and
\[  \sum_{x \in P} w(x) \geq F(P') +1, \]
as required.
If $w'(p)=+1$, and $q,r$ is not a negative pair, the above inequality holds,
as required.
If $w'(p)=+1$, and $q,r$ is a negative pair, switch both $\ell_1$ and $\ell_2$ to obtain the new weight assignment $w''$;
we have $\sum_{x \in T} w''(x) =3$, which is more than required, and
the bound follows.
Note that these (possible) switches of $\ell_1$ or $\ell_2$ do not affect $P'$ and the discrepancies
add up as in case (i).
\end{proof}

\begin{lemma} \label{lem:basis}
  Let $P$ be any set of $n \geq 3$ noncollinear points in the plane.

  If $3 \leq n \leq 5$, then $F(P) \geq n-2$. If $n=6$, then $F(P) \geq n-4 =n/3= 2$.
\end{lemma}
\begin{proof}
  Let $n=3$. Then $P$ is in general position and the claimed inequality follows from
  Lemma~\ref{lem:gen-pos}.

  Let $n=4$. Use Procedure \texttt{N} on input $P$.
  Assume first that $P$ has three collinear points on some line $\ell$, and let $p$ be the
  fourth point.
  If  $\sum_{x \in P \cap \ell} w(x) =3$, we are done, whereas if
  $\sum_{x \in P \cap \ell} w(x) =1$, then switch the line connecting $p$ with the negative point on $\ell$,
  and conclude with $\sum_{x \in P} w(x) \geq 3-1=2=n-2$, as required.
  If $P$ is in general position, the claimed inequality follows from Lemma~\ref{lem:gen-pos}.

  Let $n=5$. Use Procedure \texttt{N} on input $P$.
  We distinguish two cases.
  
  (i) Assume first that $P$ has four collinear points on a line $\ell$, and let $p$ be the remaining point.
  If there is a negative pair of points on $\ell$, perform a ``double switch'' involving $p$ and
  the negative pair. If there is only one negative point on $\ell$,
  switch the line connecting $p$ with the negative point.
In either case, we deduce that $\sum_{x \in P} w(x) \geq 4-1=3=n-2$, as required.

(ii) Assume now that $P$ has exactly three collinear points on a line $\ell$ (\ie,
no four points of $P$ are collinear). Observe that $P$ has no negative pair, since
the line incident to such a pair would be negative, a contradiction.
Therefore, we obtain that $\sum_{x \in P} w(x) \geq 4-1=3=n-2$, as required.

Let $n=6$ and refer to Fig.~\ref{fig:six}.

\begin{figure}[htbp]
\centering
\includegraphics[scale=0.8]{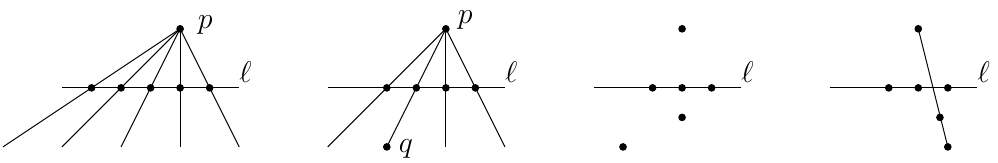}
\caption{The case $n=6$.}
\label{fig:six}
\end{figure}

Assume first that $P$ has five collinear points on a line $\ell$, and let $p$ be the remaining point.
If there is a negative pair of points on $\ell$, perform a ``double switch'' involving $p$ and
the negative pair. If there is only one negative point on $\ell$,
switch the line connecting $p$ with the negative point.
In either case, we deduce that $\sum_{x \in P} w(x) \geq 5-1=4=n-2$, as required.

Assume next that $P$ has four collinear points, $a,b,c,d$ on a line $\ell$,
and let $p$ and $q$ be the remaining points not on $\ell$. Then either $pa$ and $pb$ are ordinary lines
sharing a common point $p$, or  $qc$ and $qd$ are ordinary lines sharing a common point $q$, or both.
Use induction and Lemma~\ref{lem:reduction} to obtain
\[ F(P) \geq  (n-3)/3 +1 = n/3, \]
as required.

Assume next that $P$ has three collinear points, $a,b,c$ on a line $\ell$,
and let $p$, $q$, and $r$ be the other three points not on $\ell$.
Partition $P$ into $T=\{a,b,c\}$ and $P'=P \setminus T$.
If $P'$ is not collinear use induction on $P'$ and switch $\ell$ if needed to make it positive.
Then $F(P') \geq (n-3)/3$ and
\[ F(P) \geq  (n-3)/3 +1 = n/3, \]
as required.
If  $P'$ is collinear, use Procedure \texttt{N} on input $P$; and if there is a negative pair,
switch the ordinary line through this pair. We have
  \[  \sum_{x \in P} w(x) \geq n-2 > n/3, \]
  as required.
\end{proof}

The case of $P$ containing a line incident to many points of the set is treated in the following:

\begin{lemma} \label{lem:long-line}
  Let $P$ be a noncollinear set of $n \geq 7$ points in the plane, in which $\ell$ is a
  connecting line with at least $n-3$ points. Then we have $F(P) \geq n/3$.
\end{lemma}
\begin{proof}
  If $|\ell \cap P|=n-1$, let $p$ be the point of $P$ not on $\ell$.
  Make every point $x \in P \cap \ell$ positive by switching the pair $\{p,x\}$ if needed.
  Then $F(P) \geq n-2 \geq n/3$.

  We may assume subsequently that $|\ell \cap P| \leq n-2$,
  Since $|\ell \cap P| \geq n-3 \geq 4$, there exist two points $q,r \in \ell$, so that
  $\ell_1=\ell(p,q)$ and  $\ell_2=\ell(p,r)$ are ordinary lines sharing a common point $p$,
  Moreover, $P'=P \setminus \{p,q,r\}$ is noncollinear.
Use induction and Lemma~\ref{lem:reduction}\,(i) to obtain
\[ F(P) \geq  (n-3)/3 +1 = n/3, \]
as required.
\end{proof}

\medskip
Consider all lines spanned by a set of $n$ points in the plane. For any $k \geq 2$,
let $t_k$ denote the number of lines incident to exactly $k$ points in $P$.
Double counting the number of pairs of points gives $\sum_{k=2}^n \binom{k}{2} t_k = \binom{n}{2}$.
Two key facts in our argument are the so-called  \emph{Erd\H{o}s--Purdy inequality}~\eqref{eq:erdos-purdy}
and the stronger \emph{Hirzebruch's inequality}~\eqref{eq:hirzebruch} below.
For the technical reason of properly handling a few small values of $n$,
we in fact use this latter inequality in our proof, but we could use the former one as well instead,
for larger $n$ in the induction. For the former see~\cite{EP78}, or~\cite[Ch.~7.3]{BMP05}.
For the latter see~\cite{Hir86}, or~\cite[Ch.~7.3]{BMP05}.
\begin{align}
  \max (t_2,t_3) &\geq n-1, \quad \text{ if } n \geq 25 \text{ and } t_n=0.  \label{eq:erdos-purdy} \\
  t_2 + \frac34 t_3 &\geq n + \sum_{k \geq 5} (2k-9) t_k,
  \quad \text{ if } t_n=t_{n-1} = t_{n-2} =0.  \label{eq:hirzebruch}
\end{align}
The latter inequality immediately implies $t_2 + \frac34 t_3 \geq n$ provided $t_n=t_{n-1} = t_{n-2} =0$.

\subparagraph{Proof of Theorem~\ref{thm:old}.}
We proceed by induction on $n$. The basis of the induction is $3 \leq n \leq 6$; it is covered
by Lemma~\ref{lem:basis}. Note that $n-2 \geq n/3$ in the assumed range.

Let now $n \geq 7$.  Assume first that there is a connecting line with at least $n-3$ points.
Then $F(P) \geq n/3$ by Lemma~\ref{lem:long-line}, as required.

We may subsequently assume that there is no connecting line with at least $n-3$ points,
\ie, $t_n=t_{n-1} = t_{n-2} = t_{n-3} =0$.

(i)~Suppose first that there is a connecting line $\ell$ with exactly three points of $P$, say, $T=\{p,q,r\}$.
Partition $P$ into $T$ and $P'=P \setminus T$. Note that $|P'|= n-3$, and so $P'$ is noncollinear
by our assumptions.
We apply induction and Lemma~\ref{lem:reduction}\,(i) to obtain
\[  \sum_{x \in P} w(x) \geq  (n-3)/3 +1 = n/3, \]
\ie, the discrepancies add up.

(ii)~Suppose next that there is no connecting line in $P$ with exactly three points, \ie, $t_3=0$.
By Inequality~\eqref{eq:hirzebruch}, we have $t_2 \geq n$, where $t_2$ is the number of ordinary lines
determined by $P$. Observe that there are at most $n/2$ such lines with pairwise distinct incident points.
Since $n> n/2$, there exist two ordinary lines, say, $\ell_1=\ell(p,q)$ and $\ell_2=\ell(p,r)$,
sharing a common point $p \in P$ ($p,q,r \in P$). Let $T=\{p,q,r\}$.
Partition $P$ into $T$ and $P'=P \setminus T$. Note that $|P'|= n-3$, and so $P'$ is noncollinear
by our assumptions.
We apply induction and Lemma~\ref{lem:reduction}\,(ii) to obtain
\[  \sum_{x \in P} w(x) \geq  (n-3)/3 +1 = n/3, \]
\ie, the discrepancies add up as before.
\qed

\section{Remarks}\label{sec:remarks}

\textbf{1.} Using the weak form~\cite{Be83} of Dirac's conjecture~\cite{Dir51}
(see also~\cite{Han17,PW14}), one can obtain a weaker variant of Theorem~\ref{thm:old}.
Motivated  by the Sylvester–Gallai theorem on the existence of ordinary lines~\cite{BM90,Chv21,GT13,Zee18}, 
de Zeeuw~\cite{BVZ16} asked the following question. Does  there exist a constant
$c$ such that every finite point set $P$ in the plane which cannot be covered by
two lines has three noncollinear elements such that each of the three lines
induced by them contains at most $c$ points of $P$? 
One would call such lines $c$-ordinary, and the triangle $c$-ordinary, where $c \geq 2$. 
Fulek~\etal~\cite{FMN+17} gave a positive answer (with $c=12000$).
Dubroff~\cite{Dub18} reduced this bound significantly (to $c=11$). 
On the other hand, it is known that $c \geq 3$ is needed in general. 
Note that if one can find a triangle, all three sides of which induce ordinary lines,
then one can  perform the induction step by simply deleting these three points and applying
the induction hypothesis to the remaining point set followed by a few switches
applied to the three ordinary lines. Similarly, if one can find a $c$-ordinary triangle,
then one can hope to perform a similar induction step under favorable conditions.

\smallskip
\textbf{2.} The Gale--Berlekamp switching game can be regarded as a problem in coding theory.
A configuration is a 0-1 vector $c$ of length $n$, where each coordinate corresponds to a light.
If the light is on, we write a 0, and if it is off, we write a 1. A switch vector $s$ is a 0-1 vector
with 1's at the positions corresponding to the lights we want to switch. Let $s_1,\ldots, s_k$ be
the switch vectors, and let $S$ be the \emph{linear code} consisting of all partial sums over
$\mathbb{Z}_2$. (These partial sums are the code words.) The covering radius of the code
is the minimum number $r=r(S)$ such that every configuration $c$ is at a Hamming distance
at most $r$ from one of the elements of $S$. If $c$ is at a distance $\delta$ from a code word,
it means that starting from the configuration $c$, we can switch on all but $\delta$ lights.
For the computational complexity of the Gale--Berlekamp code, see~\cite{KS09, RV08}.
A generalization of the coding approach to signed graphs was introduced and studied by
Sol\'e and Zaslavski~\cite{SZ94}.

\smallskip
\textbf{3.} Brualdi and Meyer~\cite{BrM} studied the variant of the switching game where each light
has $k$ different positions that change cyclically modulo $k$. The same problem for higher dimensional
arrays was addressed in~\cite{ArP}. See also Schauz~\cite{Sch12}.

\smallskip
\textbf{4.} It is equally natural to consider the version of the Gale--Berlekamp game where we want
to \emph{minimize} the absolute value of the discrepancy of the configuration, that is, to balance
the numbers of on and off lights as much as possible. Leo Moser conjectured that on a $k\times k$ array,
for any initial configuration, there is a sequence of row and column switches that yields
a configuration with discrepancy 0, 1, or 2. This was proved by Koml\'os and Sulyok~\cite{KoS}
for all sufficiently large $k$ and then by Beck and Spencer~\cite{BS83} for all $k$.
Similarly, one can ask: What is the minimum absolute value of the discrepancy that can be achieved
(for any initial configuration) in our purely geometric variant? The proof of the next result is 
simpler than the proofs of Theorems~\ref{thm:new} and~\ref{thm:old}.

\begin{proposition}
    Any initial configuration $(P,w_0)$ on a noncollinear point set $P$ can be
    transformed into a configuration $(P,w)$ with $|w|\in \{-2,-1,0,1,2\}$ after
    $O(|P|)$ line switches in the geometric Gale--Berlekamp game. 
\end{proposition}
\begin{proof}
    We do induction on the cardinality of $P$. The base case where $|P|=3$
    is obvious. For the inductive process, we first apply the Sylvester--Gallai
    theorem to find two points $p,q \in P$ that determine an ordinary line. Then
    we consider the point set $P\setminus \{p,q\}$. If $P \setminus \{p,q\}$ is
    a collinear set, then without loss of generality we can assume $p$ is not on
    this line, and we can apply line switches using lines passing through $p$ to
    make the total weight of $P\setminus \{p,q\}$ equal to either $0$ or $1$. If
    $P \setminus \{p,q\}$ is not collinear, then we can apply the induction
    hypothesis to $P\setminus \{p,q\}$. In either case, we can achieve a
    configuration $w'$ such that the total weight, denoted as $\omega$, of $P
    \setminus \{p,q\}$ is in $\{-2,-1,0,1,2\}$. Now, if $w'(p)+w'(q)$ has the
    same sign as $\omega$, switch the line $pq$. Then the total weight of $P$
    will still be in $\{-2,-1,0,1,2\}$, concluding the proof. 
\end{proof}


\begin{thebibliography}{99}

\bibitem {AS16}
Noga Alon and Joel Spencer,
\emph{The Probabilistic Method},  fourth edition,
Wiley, New York, 2016.

\bibitem{ArP}
Gustavo Ara\'ujo and Daniel Pellegrino,
A Gale–Berlekamp permutation-switching problem in higher dimensions,
\emph{European Journal of Combinatorics} \textbf{77} (2019), 17--30.

\bibitem{Be83}
J{\'o}zsef Beck,
On the lattice property of the plane and some problems of
Dirac, Motzkin and Erd\H{o}s in combinatorial geometry,
\emph{Combinatorica}
\textbf{3} (1983), 281--297.

\bibitem{BS83}
J\'ozsef Beck and Joel Spencer,
Balancing matrices with line shifts,
\emph{Combinatorica}
\textbf{3(3-4)} (1983), 299--304.

\bibitem{BM90}
Peter Borwein and William O. J. Moser,
A survey of Sylvester's problem and its generalizations,
\emph{Aequationes Mathematicae}
\textbf{40(1)} (1990), 111--135.

\bibitem{BVZ16}
Thomas Boys, Claudiu  Valculescu, and Frank de Zeeuw,
On the number of ordinary conics,
\emph{SIAM Journal on Discrete Mathematics}
\textbf{30(3)} (2016), 1644-1659.

\bibitem{BMP05}
Peter Bra\ss , William Moser, and J\'anos Pach,
\emph{Research Problems in Discrete Geometry},
Springer, New York, 2005.

\bibitem{BS71}
Thomas A. Brown and Joel H. Spencer,
Minimization of $\pm1$ matrices under line shifts,
\emph{Colloquium Mathematicum}
\textbf{1(23)} (1971), 165--171.

\bibitem{BrM}
Richard A. Brualdi and Seth A. Meyer,
A Gale–Berlekamp permutation-switching problem,
\emph{European Journal of Combinatorics} \textbf{44} (2015), 43--56.

\bibitem{CaS}
Jordan Carlson and Daniel Stolarski,
The correct solution to Berlekamp's switching game,
\emph{Discrete Mathematics} \textbf{287(1-3)} (2004), 145--150.

\bibitem{Chv21}
Va\v{s}ek Chv\'atal,
\emph{The Discrete Mathematical Charms of Paul Erd\H{o}s},
Cambridge University Press, New York, 2021.

\bibitem{Dir51}
Gabriel A.~Dirac,
Collinearity properties of sets of points,
\emph{The Quarterly Journal of Mathematics}
\textbf{2(1)} (1951), 221--227.

\bibitem{Dub18}
Quentin Dubroff,
A better bound for ordinary triangles,
Preprint, 2018, \url{arXiv:1805.06954}.

\bibitem{EP78}
Paul Erd\H{o}s and George Purdy,
Some combinatorial problems in the plane,
\emph{Journal of Combinatorial Theory Ser. A}
\textbf{25} (1978), 205--210.

\bibitem{FS89}
Peter C. Fishburn and Neil J. Sloane,  
The solution to Berlekamp's switching game,
\emph{Discrete Mathematics}
\textbf{74(3)} (1989), 263--290.

\bibitem{FMN+17}
Radoslav Fulek, Hossein N. Mojarrad, M{\'a}rton Nasz{\'o}di, J{\'o}zsef Solymosi,
Sebastian Stich, and May Szedl{\'a}k,
On the existence of ordinary triangles,
\emph{Computational Geometry: Theory and Applications}
\textbf{66} (2017), 28--31.

\bibitem{GT13}
Ben Green and Terence Tao,
On sets defining few ordinary lines,
\emph{Discrete \& Computational Geometry}
\textbf{50(2)} (2013), 409--468.

\bibitem{Han17}
 Zeye Han,
A Note on the weak Dirac conjecture,
\emph{Electron. J. Comb.}
\textbf{24(1)} (2017), \#P1.63.

\bibitem{Hir86}
Friedrich Hirzebruch,
Singularities of algebraic surfaces and characteristic numbers,
\emph{Contemporary Mathematics}
\textbf{58} (1986), 141--155.

\bibitem{KS09}
Marek Karpinski and Warren Schudy,
Linear time approximation schemes for the Gale--Berlekamp game and related minimization problems,
\emph{Proc. 41st annual ACM Symp. Theory of Comp. (STOC 2009)}, 313--322.

\bibitem{KoS}
J\'anos Komlós and Mikl\'os Sulyok, 
On the sum of elements of ±1 matrices, 
in: \emph{Combinatorial Theory and Its Applications (Erd\H os et al., eds.)}, 
North-Holland 1970, 721--728.

\bibitem{PW14}
Michael S. Payne and David R. Wood,
Progress on Dirac's conjecture,
\emph{Electron. J. Comb.}
\textbf{21(2)} (2014), \#P2.12.

\bibitem{RV08} 
Ron M. Roth and Krishnamurthy Viswanathan,
On the hardness of decoding the Gale–Berlekamp code,
\emph{IEEE Transactions on Information Theory} \textbf{54(3)} (2008), 1050--1060.

\bibitem{Sch12}
Uwe Schauz,
Anti-codes in terms of Berlekamp's switching game.
\emph{Electron. J. Combin.}
\textbf{19} (2012), \#P10.

\bibitem{SZ94}
Patrick Sol\'e and Thomas Zaslavsky,
A coding approach to signed graphs,
\emph{SIAM Journal on Discrete Mathematics} \textbf{7(4)} (1994), 544--553.

\bibitem{Spe94}
Joel Spencer,
\emph{Ten Lectures on the Probabilistic Method},
SIAM, 1994.

\bibitem{Zee18}
Frank de Zeeuw,
Spanned lines and Langer's inequality,
Preprint, 2018, \url{arXiv:1802.08015}.

\end{thebibliography}
\end{document}